\documentclass[12pt]{amsart}
\usepackage{amsmath, amsthm,hyperref}

\newtheoremstyle{theorem}
  {10pt}		  
  {10pt}  
  {\sl}  
  {}     
  {\bf}  
  {. }    
  { }    
  {}     
\theoremstyle{theorem}
\newtheorem{thm}{Theorem}

\newtheorem{lem}{Lemma}
\newtheorem{conj}{Conjecture}
\newtheorem{rmk}{Remark}
\newtheoremstyle{defi}
  {10pt}		  
  {10pt}  
  {\rm}  
  {}     
  {\bf}  
  {. }    
  { }    
  {}     
\theoremstyle{defi}

\newcommand{\R}{\operatorname{\mathbf{R}}}
\newcommand{\ddef}{\stackrel{\text{def}}{=}}
\newcommand{\sgn}{\mathsf{sgn}}
\newcommand{\dd}{\mathrm{d}}
\newcommand{\A}{\mathcal{A}}

\begin{document}

\title{Towards A Universal Gibbs Constant}

\author{Santanu Antu$^1$ AND John Cullinan$^2$\\
$^1$Department of Mathematics, Bard College, Annandale-On-Hudson, NY 12401\\[2pt]
$^2$Department of Mathematics, Bard College, Annandale-On-Hudson, NY 12401
}

\maketitle

\begin{abstract}
In this paper we build on the work of \cite{kaber} where it was shown that the one-parameter family of Gegenbauer Polynomials (GP) exhibit a Gibbs Phenomenon at a jump discontinuity.  We show that the one-parameter family of Generalized Laguerre Polynomials (GLP) also exhibit a Gibbs Phenomenon.  Among many differences, a major one is that the GLP are orthogonal on a non-compact subset of $\R$, while the GP are orthogonal on $[-1,1]$.  Our strategy follows that of \cite{kaber} and we use entirely elementary methods to arrive at our result.  As a special case we show that the Hermite Polynomials also possess a Gibbs Phenomenon.  We conclude with a numerical example exhibiting the rate of convergence to the Gibbs constant and a conjectured identity for special values of the GLP.


\end{abstract}

\section{Introduction}

\subsection{Background and Setup} The Gibbs Phenomenon is a well known consequence of approximating a jump-discontinuous function $f$ by a Fourier series.  Near the discontinuity (say at $x_0$), if the function jumps by height $h$ (\emph{i.e.}~$\lim_{\epsilon_1,\epsilon_2 \to 0} |f(x_0 + \epsilon_1) - f(x_0-\epsilon_2)| = h$, then the Fourier series $\mathcal{F}(f)$ of $f$ over-and-undershoots the limit by $\approx 18\%$ of $h$.  The Gibbs constant $\gamma$ is the value of the sinc integral
\[
\gamma \ddef \frac{2}{\pi} \int_0^\pi \frac{\sin t}{t}\,{\rm d}t  = 1.179\dots
\]
and is the exact proportion of the over/undershoot.

In \cite{kaber}, the author shows that a certain one-parameter family of orthogonal polynomials exhibits the same Gibbs phenomenon with the identical constant.  That is, if one expands a function near a jump discontinuity using this family of orthogonal polynomials (in contrast to a classical Fourier series), then it overshoots by exactly the same proportion $\gamma$.  

The aim of this paper is to generalize the results of \cite{kaber} to other families of orthogonal polynomials, notably ones orthogonal on a non-compact subset of the real line.  In particular, we focus on the Generalized Laguerre Polynomials (GLP) and, by specialization, the Hermite Polynomials.  The GLP are a one-parameter family of polynomials orthogonal on $\R_{\geq} \ddef [0,\infty)$; they share many of the salient features of the polynomials studied in \cite{kaber}, but differ in key ways as well.

Saving a more detailed exposition of the previous results for Section \ref{background}, we are content to give a brief recap of the main results of \cite{kaber}. In that paper, the author considers the family of Gegenbauer Polynomials\footnote{In \cite{kaber}, the author refers to the $C_n^{(\lambda)}$ as {Jacobi polynomials} but acknowledges the ambiguity in doing so. We reserve the term ``Jacobi Polynomial" for the usual two-parameter family.}  $\lbrace C_n^{(\lambda)} \rbrace_{n=0}^\infty$, orthogonal on $[-1,1]$, for all $\lambda > -1/2$.  The  Gegenbauer Polynomials (GP) are a one-parameter subfamily of the two-parameter family $\lbrace P_{n}^{(\alpha,\beta)} \rbrace_{n=0}^\infty$ of Jacobi Polynomials.  For various values of $\lambda$, the GP encompass important classes of orthogonal polynomals, most notably the Legendre Polynomials and the Chebychev Polynomials of the first and second kind.  

Among the many beautiful results of \cite{kaber}, we highlight the two main ones here.  First, the author shows that, for all $\lambda > -1/2$  the Gibbs phenomenon holds for the $\sgn$ function on $[-1,1]$; that is, the expansion of $\sgn(x)$ in the GP converges to $\gamma$ as $x \to 0$.  Second, the author considers the \emph{steepness} of approximation in the GP basis: how the derivative at $x=0$ of the approximation varies with $\lambda$.  They show that, for fixed degree $n$, the steepness increases as $\lambda$ increases.  In terms of approximation, this shows that for fixed degree $n$, larger values of $\lambda$ give better approximations to $\gamma$ closer to $x=0$.

Turning to the business of this paper, we use the same approach to study the GLP, defined for $\alpha > -1$ by
\[
L_n^{(\alpha)}(x) = \sum_{j=0}^n \binom{n+\alpha}{n-j} \frac{(-x)^j}{j!}.
\]
The $\lbrace L_n^{(\alpha)}\rbrace_{n=0}^\infty$ are orthogonal with respect to the weight function $w^{(\alpha)}:\R_{\geq} \to \R_{\geq}$ defined by 
\[
w^{(\alpha)}(x) = x^\alpha e^{-x}.
\]
As remarked above, the non-compactness of the interval of orthogonality of the GLP plays a role in the analysis. 

In general, if $\lbrace p_j \rbrace_{j=0}^\infty$ is a family of orthogonal polynomials on an interval (say $[-1,1]$), let $\mathcal{S}_n$ be the $n$th approximation of (\emph{e.g.})~the $\sgn$ function in the $p_j$:
\[
\mathcal{S}_n(x) = \sum_{j=0}^n \hat{s}_j\, p_j(x).
\]
Let $x_{n,+}$ be the smallest positive critical point of $\mathcal{S}_n$ and $x_{n,-}$ the largest negative critical point of $\mathcal{S}_n$.  By a ``Gibbs Phenomenon" for  $\lbrace p_j \rbrace_{j=0}^\infty$ we mean that 
\[
\lim_{n \to \infty} S_n(x_{n,+})  - S_n(x_{n,-}) = \gamma.
\]
This necessitates an analysis of the critical points of the expansion.

In \cite{kaber}, the author shows that the critical points of the Gegenbauer approximation are the exactly the zeroes of a different Gegenbauer polynomial.  The zeroes of every Gegenbauer polynomial lie in the interval $[-1,1]$, and their distribution is such that as the degree of the polynomial increases, the smallest positive zero approaches 0.   This is an oversimplification of the delicate argument of \cite[\S3.3]{kaber}, but gives a rough justification that the \emph{shape} of the approximation ought to resemble the classical Fourier approximation of the $\sgn$ function.  That the \emph{value} of the approximation at the critical point approaches $\gamma$ is an entirely separate argument.  We use the same approach in the current paper to prove the following theorem.

\begin{thm} \label{mainthm}
Let $S:[0,\infty) \to \R$ be defined by $S(x) = 0$ for $0 \leq x  < 1$ and $S(x) = 1$ for $x \geq 1$.  Let 
\[
\Pi_n^{(\alpha)}(S)(x) = \sum_{j=0}^n \hat{s}_j^{(\alpha)}\, L_j^{(\alpha)}(x)
\]
be the expansion of $S$ in the Generalized Laguerre Polynomials $\lbrace L_j^{(\alpha)} \rbrace_{j=0}^\infty$ for $\alpha > -1$.  Let $x_{n,+}^{(\alpha)}$ be the smallest critical point of $\Pi_n^{(\alpha)}(S)$ that is strictly greater than $1$ and $x_{n,-}^{(\alpha)}$ the largest critical point strictly smaller than $1$. Then
\[
\lim_{n \to \infty} \Pi_n^{(\alpha)}(S)(x_{n,+}^{(\alpha)})  - \Pi_n^{(\alpha)}(S)(x_{n,-}^{(\alpha)}) = \frac{2}{\pi} \int_0^\pi \frac{\sin t}{t} \, \dd t.
\]
\end{thm}

The Hermite polynomials (see Section \ref{background} below) are special cases of the GLP, obtained by setting $\alpha = \pm 1/2$ and evaluating at $x^2$; as such they are orthogonal on $\R$ rather than just $\R_{\geq}$.  Because of this difference we treat them separately from the GLP and show that they have a Gibbs Phenomenon as well.  We note that the Hermite Polynomials are even and odd functions in even and odd degrees, respectively, hence we can deduce a Gibbs Phenomenon by simply looking at one of the critical points (we choose $x_{n,+}$).

\begin{thm} \label{hermthm}
Let $\lbrace H_j \rbrace_{j=0}^\infty$ be the family of Hermite polynomials, let 
\[
\Pi_n(\sgn)(x) = \sum_{j=0}^n \hat{s}_j H_j(x)
\]
be the expansion of the $\sgn$ function, and let $x_{n,+}$ be the smallest positive critical point of $\Pi_n(\sgn)$.  Then
\[
\lim_{n \to \infty} \Pi_n(\sgn)(x_{n,+}) = \frac{2}{\pi} \int_0^\pi \frac{\sin t}{t} \, \dd t.
\]
\end{thm}

We expect that a more general result is true: namely that the two parameter family (in $(\alpha,\beta)$ -- see Section \ref{background} below) of Jacobi Polynomials (JP) exhibits a Gibbs phenomenon. Similar to the GP, the JP are orthogonal on $[-1,1]$, however they are not even and odd functions in even and odd degree, respectively, making for a different analysis than in \cite{kaber}.  Indeed, one of the places in the current paper where we differ significantly is in the analysis of the Christoffel-Darboux sum (compare \cite[Prop.~1]{kaber} with (\ref{der_form}) below).  

Furthermore, we have the following relationship between the JP and GLP in which the latter can be viewed as a specialization ``at infinity'' of the former \cite[(5.3.4)]{szego}:
\begin{align} \label{lag_inf}
L_n^{(\alpha)}(x) = \lim_{\beta \to \infty} P_n^{(\alpha,\beta)}(1-2x/\beta).
\end{align}
If one could prove that the JP have a Gibbs Phenomenon for all finite specializations of $\alpha$ and $\beta$, and then use (\ref{lag_inf}) to deduce the Gibbs Phenomenon for the GLP, this would give evidence for what we term the \emph{universal Gibbs constant} $\gamma$.  There will be always be more general polynomials (\emph{e.g.}~the Hahn Polynomials yield the JP under suitable specialization in the Askey scheme) for working out special cases, but it would be interesting to determine necessary and sufficient conditions for an arbitrary family of orthogonal polynomials to possess a Gibbs Phenomenon.  Our intent in this paper is to use elementary methods to give evidence for such a ``universal Gibbs Phenomenon'' for orthogonal polynomials.

\section{Orthogonal Polynomials and Series Expansions} \label{background}

\subsection{Classical Orthogonal Polynomials} Here we give a brief recap of some of the main examples of orthogonal polynomials on subsets of the real line.   The \textbf{Jacobi polynomials} $\lbrace P_n^{(\alpha,\beta)}(x)\rbrace_{n=0}^\infty$ are a two-parameter family of polynomials on the interval $[-1,1]$, orthogonal with respect to the weight function $w^{(\alpha,\beta)}(x) = (1-x)^\alpha(1+x)^\beta$,  and defined explicitly by 
\[
P_n^{(\alpha,\beta)}(x) = \sum_{j=0}^n \binom{n+a}{n-j} \, \binom{n+\beta}{j} \, \left( \frac{x-1}{2}\right)^j \, \left( \frac{x+1}{2} \right)^{n-j}.
\]
Many classical families of orthogonal polynomials can be viewed as special cases of Jacobi polynomials, such as
\begin{itemize}
\item the \textbf{Gegenbauer polynomials} $C_n^{(\lambda)}(x) = P_n^{(\lambda-1/2,\lambda-1/2)}(x)$,
\item the \textbf{Legendre polynomials} $P_n(x) = P_{n}^{(0,0)}(x)$,
\item the \textbf{Chebyshev polynomials} $T_n(x) = P_n^{(-1/2\, ,\,-1/2)}(x)$ of the first kind and $U_n(x) = P_n^{(1/2,1/2)}(x)$ of the second kind, and
\item the \textbf{Generalized Laguerre polynomials} 
\[
L_n^{(\alpha)}(x) = \lim_{\beta \to \infty} P_n^{(\alpha,\beta)}(1-2x/\beta).
\]
\end{itemize}
While the first three families are orthogonal on $[-1,1]$ as well, the $L_n^{(\alpha)}(x)$ are orthogonal on the non-negative real numbers $[0,\infty)$.   The \textbf{Hermite polynomials} 
\begin{align*}
H_{2n}(x) &=(-1)^n\,2^{2n}\,n!\,L_n^{(-1/2)}(x^2) \\
H_{2n+1}(x) &=(-1)^n\,2^{2n+1}\,n!\,x\,L_n^{(1/2)}(x^2),
\end{align*}
which we will work with in detail below, are specializations of the Generalized Laguerre polynomials and are orthogonal on $(-\infty,\infty)$.

\subsection{Overview of Kaber's Proof}

Even though the main result in \cite{kaber} holds for the family Gegenbauer polynomials (and therefore all of its specializations, such as the Legendre polynomials), the approach would work more generally in the following context.

Let $\mathcal{P} = \lbrace p_n(x) \rbrace_{n=0}^{\infty}$ be a family of polynomials, orthogonal with respect to a weighted inner product
\[
\langle \ , \ \rangle: \mathcal{P} \times \mathcal{P} \to \R
\]
on a subset $S$ of the real line.  Let $f: S \to \R$ be a bounded real-valued function on $S$.  We consider the series expansion $\mathcal{F}_n$ of $f$ in the polynomials $p_k$:
\[
\mathcal{F}_n(x) = \sum_{k=0}^n \hat{f}_k\, p_k(x),
\]
where $\hat{f}_k$ is the orthogonal projection of $f$ onto $p_k$.  We set the pointwise limit $\mathcal{F}(x) = \lim_n \, \mathcal{F}_n(x)$ when it exists.  

In \cite{kaber}, the author takes advantage of several properties enjoyed by the Gegenbauer polynomials (but also by many families of orthogonal polynomials) namely:
\begin{enumerate}
\item the $C_n^{(\lambda)}(x)$ are odd functions in odd degree and even functions in even degree, and
\item the function $f = \sgn$ is odd, and
\item the $C_n^{(\lambda)}(x)$ satisfy a Sturm-Liouville differential equation.
\end{enumerate}
All of this allows for explicit calculation \cite[p.~140]{kaber}.  Writing $\mathcal{F}_n^{(\lambda)}(x)$ for the expansion of $\sgn$ in the the polynomials $p_n = C_n^{(\lambda)}$, we get 
\[
\mathcal{F}_{2n+1}^{{(\lambda)}}(x) = \frac{1}{\lambda} \sum_{k=0}^n \frac{C_{2k}^{(\lambda+1)(0)}}{\| C_{2k}^{(\lambda+1)} \|} C_{2k+1}^{(\lambda)}(x).
\]
Differentiating, using the fact that that $\left( C_{2k+1}^{(\lambda)}(x) \right)' = 2\lambda C_{2k}^{(\lambda+1)}$,  and applying the Christoffel-Darboux summation formula  \cite[(3.2.3)]{szego}, we get
\[
\frac{\dd}{\dd x} \mathcal{F}_{2n+1}^{(\lambda)}(x) = d_{n}^{(\lambda)}  \, \frac{C_{2n+1}^{(\lambda+1)}(x)}{x},
\]
for an explicit constant $d_{n}^{(\lambda)}$. 

This last expression is especially convenient.  It shows that the critical points of $\mathcal{F}_{2n+1}^{(\lambda)}$ are given by the non-zero zeroes of a shifted Gegenbauer polynomial.  Evaluating $\mathcal{F}_{2n+1}^{(\lambda)}$ at its smallest positive critical point (say, $x_{n,+}$) is the approximation of the Gibbs phenomenon and the limit $\lim_{n \to \infty}\mathcal{F}^{(\lambda)}_{2n+1}(x_{n,+})$ exists and is exactly the Gibbs constant $\gamma$.  Appealing again to the the derivative expression above, write
\[
\mathcal{F}_{2n+1}^{(\lambda)}(x_{n,+}) = d_{n,\lambda} \int_0^{x_{n,+}} \frac{C_{2n+1}^{(\lambda+1)}(y)}{y}\, \dd y.
\]
It is this final expression that is shown to converge to $\gamma$, for all $\lambda$, as $n \to \infty$.  The argument is intricate and uses a very fine analysis of the zeroes of the Gegenbauer polynomials.

In the following sections we will apply this exact approach to the Generalized Laguerre polynomials.  However, several complications arise:
\begin{enumerate}
\item the GLP are not orthogonal when $x <0$, hence we consider the expansion of 
\[
S(x) = \begin{cases} 0 & \text{ if } 0\leq x <1 \\ 1 & \text{ if } x \geq 1; \end{cases}
\]
\item we have $L_n^{(\alpha)}(1) \ne 1/2$, unlike the analogous situation with the Gegenbauer polynomials and the $\sgn$ function;
\item the $L_n^{(\alpha)}(x)$ are neither even functions in even degree, nor odd functions in odd degree, and do not exhibit a similar symmetry about $x =1$.  
\item since $L_n^{(\alpha)}(1) \ne 1/2$ for all $n$ and $\alpha$ (unlike the situation where $C_{2k+1}^{(\lambda)}(0) =0$ for all $k$ and $\lambda$ , we must consider the entire jump $\Pi_n^{(\alpha)}(S)(x_{n,+}^{(\alpha)})  - \Pi_n^{(\alpha)}(S)(x_{n,-}^{(\alpha)})$ and show that this converges to $\gamma$;
\item the domain $\R_\geq$ on which that the polynomials $\lbrace {L}_n^{(\alpha)}\rbrace$ are orthogonal is not compact, so we will need an alternative argument to show that certain critical points of the expansion of $f$ in the GLP get arbitrarily close to $x=1$.  
\end{enumerate}

\begin{rmk}
Towards point (4) above, we make the conjecture that for all $\alpha$, $\lim_{n \to \infty} \Pi_n^{(\alpha)}(1) = 1/2$; see Section \ref{numerics} for details.  If true, then we still recover Theorem \ref{mainthm} but could deduce the Gibbs Phenomenon from either of the one-sided limits.
\end{rmk}




\section{Special Case: The Hermite Polynomials}

As a special case, we briefly work through the details for the Hermite polynomials.  It's possible we could derive the Gibbs constant as a corollary to our work on the GLP below, but because the Hermite polynomials involve two specializations ($\alpha = \pm 1/2$ and $x \to x^2$) and, due to the quadratic substitution, the domain of orthogonality is $(-\infty,\infty)$.   The Hermite polynomials are also simpler to compute with since they are odd functions in odd degree and even functions in even degree.  

We write $\Pi_M$ for the Hermite approximation of $\sgn$ of degree $M$. Therefore, $\Pi_M(x) = \sum_{j=0}^M \hat{s}_jH_j(x)$, where 
\begin{align*}
\hat{s}_j = \frac{\langle \sgn, H_j \rangle}{\| H_j \|^2} &= \frac{1}{\| H_j \|^2} \int_{-\infty}^\infty \sgn(x)\, H_j(x) \, e^{-x^2}\,{\rm d}x \\[5pt] &= \begin{cases} 0 & \text{ if $j$ even} \\ \displaystyle \frac{2}{\| H_j \|^2} \int_{0}^\infty  H_j(x) \, e^{-x^2}\,{\rm d}x & \text{ if $j$ odd.}
\end{cases}
\end{align*}
To evaluate the integral, use the Sturm-Liouville form of the differential equation satisfied by the Hermite polynomials:
\begin{align}
\left( e^{-x^2} \, H_j'(x) \right)' + e^{-x^2}\,2j\,H_j(x) = 0.
\end{align}
Using the fact that $H_j'(x) = 2j\,H_{j-1}(x)$, we integrate to obtain $\int_0^\infty  H_j(x)e^{-x^2}{\rm d}x = H_{j-1}(0)$ and conclude
\begin{align}
\hat{s}_j = 
\begin{cases} \displaystyle
\frac{2H_{j-1}(0)}{\|H_j\|^2} & \text{if $j$ odd, and } \\ 0 & \text{if $j$ even.} \end{cases}
\end{align}

In general, we have $\| H_k \|^2 = \sqrt{\pi}\,2^k\,k!$, whence $\|H_{2j+1}\|^2 = 2(2j+1) \,\|H_{2j}\|^2$.  Therefore, we can reindex and write the series expansion of $\sgn(x)$ in the Hermite basis as
\begin{align}
\Pi_{2N+1}(x) = \sum_{j=0}^N \left( \frac{1}{2j+1} \right) \, \frac{H_{2j}(0) H_{2j+1}(x)}{\|H_{2j}\|^2}. 
\end{align}

According to \cite[(5.5.9)]{szego}, the Christoffel-Darboux formula applied to the Hermite polynomials gives us
\[
\sum_{k=0}^n \left(2^k\,k!\right)^{-1} H_k(x)\,H_k(y) = \left(2^{n+1}\,n!\right)^{-1}\, \frac{H_{n+1}(x)H_n(y) - H_n(x)H_{n+1}(y)}{x-y},
\]
which we can use to get an expression for $\Pi_{2N+1}$ and its critical points.  We have 
\begin{align*}
\frac{\dd}{\dd x} \Pi_{2N+1}(x) &= \sum_{j=0}^N \left( \frac{1}{2j+1} \right) \, \frac{H_{2j}(0)}{\|H_{2j}\|^2}\, 2(2j+1)\,H_{2j}(x).
\end{align*}
Simplifying, reindexing, and using the formula $\| H_j \|^2 = \sqrt{\pi}\,2^j\,j!$, we can rewrite the sum above into a form amenable to the Christoffel-Darboux expression (using the fact that $H_j(0)=0$ when $j$ is odd):
\begin{align*}
\frac{\dd}{\dd x} \Pi_{2N+1}(x) &= \frac{2}{\sqrt{\pi}}\sum_{j=0}^{2N} \left( 2^j\,j!\right)^{-1} H_{j}(0)\,H_j(x) = \frac{2\,H_{2N+1}(x) \, H_{2N}(0)}{x\,\sqrt{\pi} 2^{2N+1} \, (2N+1)!}.
\end{align*}
We can further simplify this expression using the well-known value of even-degree Hermite polynomials at $x=0$: $H_{2m}(0) = (-1)^m\,(2m)!/m!$.  Substituting this into the expression above yields
\begin{align} \label{diff_form}
\frac{\dd}{\dd x} \Pi_{2N+1}(x) = \frac{(-1)^N}{2^{2N}\,N! \sqrt{\pi}}\, \frac{H_{2N+1}(x)}{x},
\end{align}
and hence the following integral form for $\Pi_{2N+1}(x)$:
\begin{align}
\Pi_{2N+1}(x) = \frac{(-1)^N}{2^{2N}\,N! \sqrt{\pi}}\, \int_0^x\,\frac{H_{2N+1}(y)}{y}\,\dd y.
\end{align}
We can now use these forms to analyze the critical points of $\Pi_{2N+1}$.

\subsection{Critical Points}

Let $x_{N,+}$ be the smallest positive critical point of $\Pi_{2N+1}$.  In order to understand the Gibbs phenomenon for the Hermite polynomials, we are interested in $\lim_{N \to \infty} \Pi_{2N+1}(x_+)$.  By (\ref{diff_form}), the critical points of $\Pi_{2N+1}$ are precisely the non-zero zeroes of $H_{2N+1}$.  

Now we appeal to \cite[(6.23.31)]{szego} on estimates for the least positive zero $x_{1n}$ of $H_n(x)$:
\[
\frac{\pi}{\sqrt{2n+1}} < x_{1n} < \frac{\pi}{\sqrt{2n+1}} \underbrace{\left\{ \frac{1}{2} + \frac{1}{2} \left[1 - \left( \frac{2\pi}{2n+1} \right)^2 \right]^{1/2}\right\}^{-1/2}}_{\to 1^+ \text{ as } N \to \infty} 
\]
Thus, for large $N$ we have the approximation 
\begin{align*}
x_{N,+} \sim \frac{\pi}{\sqrt{4N+3}}.
\end{align*}

\subsection{Asymptotic Analysis} For large $n$ we have the well-known estimate \cite[(8.22.6,7)]{szego}:
\[
e^{-x^2/2} \, H_n(x) \, \sim \, \frac{2^n}{\sqrt{\pi}} \, \Gamma \left(\frac{n+1}{2} \right) \, \cos\left( x \sqrt{2n} - \frac{n\pi}{2} \right). 
\]
Applying this to $H_{2N+1}(x)$, simplifying the Gamma factor, and rewriting the cosine function as a sine function gives us
\[
H_{2N+1}(x) \sim \frac{1}{\sqrt{\pi}}\, (-1)^N\, N!\, 2^{2N+1}\,e^{x^2/2}\,\sin\left( x\sqrt{4N+2} \right), 
\]
and hence
\begin{align*}
\Pi_{2N+1}(x_{N,+}) &\sim \frac{2}{\pi} \, \int_0^{\pi/\sqrt{4N+3}}\, \frac{e^{y^2/2}\,\sin \left(y \sqrt{4N+2} \right)}{y}\, \dd y \\
&= \frac{2}{\pi} \, \int_0^{\pi \, \sqrt{\frac{4N+2}{4N+3}}}\, \frac{\exp \left( \frac{u^2}{8N+4} \right)\,\sin \left(u \right)}{u}\, \dd u \\
&\to \frac{2}{\pi} \int_0^\pi \frac{\sin u}{u} \, \dd u = \gamma,
\end{align*}
as $N \to \infty$.  This recovers the classical Gibbs constant in the context of the Hermite polynomials.  In the next section will apply similar methods to recover the Gibbs constant for the GLP, though with some necessary modifications. 

\section{The Generalized Laguerre Polynomials}

We now turn to the main business of the paper and show that the Gibbs constant $\gamma$ appears in the approximation by Laguerre polynomials near a jump discontinuity.  Because the interval on which the GLP are orthogonal is $[0,\infty)$, we must use a different test function than $\sgn$.  Set
\[
S(x) \ddef \begin{cases} 0 &\text{ if } 0 \leq x <1 \\ 1 & \text{ if } x \geq 1. \end{cases}
\]

Proceeding similarly as above, write
\[
\Pi_n^{(\alpha)}(S)(x) = \sum_{j=0}^n \hat{s}_j^{(\alpha)} L_j^{(\alpha)}(x),
\]
where
\[
\hat{s}_j^{(\alpha)} = \frac{\langle S,L_j^{(\alpha)} \rangle}{\| L_j^{(\alpha)} \|^2} = \frac{1}{\| L_j^{(\alpha)} \|^2} \,\int_1^\infty L_j^{(\alpha)}(x)\, e^{-x}\, x^\alpha\, \dd x.
\]
We calculate the integral with help from the Sturm-Liouville differential equation
\[
\left(x^{\alpha+1}\, e^{-x} \, L_n^{(\alpha)}(x)'\right)' = -n\, x^\alpha \, e^{-x} \, L_n^{(\alpha)}(x),
\]
giving us for $j\geq 1$
\begin{align*}
\hat{s}_j^{(\alpha)} &= \frac{1}{\| L_j^{(\alpha)} \|^2}  \, \frac{-1}{je} \, \frac{\dd}{\dd x} \bigg|_{x=1} L_j^{(\alpha)}(x) \\
&= -\frac{1}{e}\, \frac{L_{j-1}^{(\alpha+1)}(1)}{\| L_{j-1}^{(\alpha + 1)} \|^2},
\end{align*}
where the last equality follows from manipulating well-known identities for the GLP; when $j=0$ we calculate directly that
\[
s_0^{(\alpha)} = \frac{\int_1^\infty x^\alpha e^{-x}\,\dd x}{\Gamma(\alpha +1)} = \frac{\Gamma(\alpha+1,1)}{\Gamma(\alpha + 1)},
\]
where $\Gamma(\alpha+1,1)$ is the incomplete Gamma function.  Putting this all together gives
\begin{align} \label{sum}
\Pi_n^{(\alpha)}(S)(x) =\frac{\Gamma(\alpha+1,1)}{\Gamma(\alpha + 1)}  -\frac{1}{e} \sum_{j=1}^n \frac{L_{j-1}^{(\alpha+1)}(1)}{\| L_{j-1}^{(\alpha + 1)} \|^2} L_{j}^{(\alpha)}(x).
\end{align}

We can simplify this expression by first taking the derivative
\begin{align*}
\frac{\dd}{\dd x} \, \Pi_n^{(\alpha)}(S)(x) &=-\frac{1}{e} \sum_{j=1}^n \frac{L_{j-1}^{(\alpha+1)}(1)}{\| L_{j-1}^{(\alpha + 1)} \|^2} \, \frac{\dd}{\dd x} \, \left(L_{j}^{(\alpha)}(x) \right) \\
&=\frac{1}{e} \sum_{j=0}^{n-1} \frac{L_{j}^{(\alpha+1)}(1) \,L_{j}^{(\alpha+1)}(x)}{\| L_{j}^{(\alpha + 1)} \|^2}, 
\end{align*}
and then using the Christoffel-Darboux summatory identity 
\[
\sum_{j=0}^n \frac{p_j(x)\,p_j(y)}{\| p_j \|^2} = \frac{k_n}{k_{n+1} \|p_{n+1}\|^2} \, \frac{p_{n+1}(x)p_n(y) - p_{n}(x)p_{n+1}(y)}{x-y}
\]
(here $k_n$ is the leading coefficient of $p_n$) to obtain
\begin{align} \label{der_form}
\frac{\dd}{\dd x} \, \Pi_n^{(\alpha)}(S)(x) = \frac{-n!}{e\, \Gamma(n+\alpha+1)} \, \left( \frac{L_{n-1}^{(\alpha+1)}(1)\, 
L_n^{(\alpha+1)}(x) - L_{n-1}^{(\alpha+1)}(x)\,L_n^{(\alpha+1)}(1)}{x-1} \right),
\end{align}
whence
\begin{align} \label{int_form}
\Pi_n^{(\alpha)}(S)(x) =  \frac{n!}{e\, \Gamma(n+\alpha+1)} \, \int_1^x \,  \frac{-L_{n-1}^{(\alpha+1)}(1)\, 
L_n^{(\alpha+1)}(y) + L_{n-1}^{(\alpha+1)}(y)\,L_n^{(\alpha+1)}(1)}{y-1} \, \dd y
\end{align}
for $x>1$. For convenience we set
\[
d_n^{(\alpha)} \ddef \frac{n!}{e\, \Gamma(n+\alpha+1)}. 
\]

The numerator of (\ref{der_form}) is a polynomial (divisible by $x-1$), so it make sense to speak of its smallest zero $>1$ and largest zero $<$ 1; we denote these zeroes by $x_{n,+}^{(\alpha)}$ and $x_{n,-}^{(\alpha)}$, respectively.  Over the next several subsections we perform an asymptotic analysis on the expression in (\ref{int_form}) in order to show that 
\[
\lim_{n \to  \infty} \Pi_n^{(\alpha)}(S)(x_{n,+}^{(\alpha)}) - \Pi_n^{(\alpha)}(S)(x_{n,-}^{(\alpha)}) = \gamma.
\]

\subsection{Initial Asymptotic Analysis}

In order to convert the integral (\ref{int_form}) into one resembling the Gibbs constant, we review the well-known asymptotic relations between the GLP and the sine function.  We start with an asymptotic expression for $d_n^{(\alpha)}$. 

\begin{lem} \label{leading}
We have $d_n^{(\alpha)} \sim 1/(e\,n^\alpha)$ for large $n$ and all $\alpha$.
\end{lem}

\begin{proof}
We use a standard approximation to estimate the Gamma function.  Fixing $\alpha \in \R$, we take $n \to \infty$ in the following analysis:
\begin{align*}
\frac{n!}{\Gamma(n+\alpha +1)} = \frac{\Gamma(n+1)}{\Gamma(n+\alpha+1)} &\sim \frac{\exp(n\log n - n)}{\exp((n+\alpha)\log (n+\alpha) - (n+\alpha))}\\
&= \frac{1}{\exp \left( (\log(n+\alpha))^{n+\alpha} - (\log n)^n + n - n -\alpha \right)} \\
&= \frac{\exp(\alpha)}{\exp \left( \log ((n+\alpha)/n))^{n} \cdot (n+\alpha)^\alpha \right)}\\
&=\frac{\exp(\alpha)}{ \left( \frac{n+\alpha}{n} \right)^n \cdot n^\alpha \cdot (1 + \frac{\alpha}{n})^\alpha} \\
&\to \frac{\exp(\alpha)}{\exp(\alpha) \cdot n^\alpha} = 1/n^\alpha,
\end{align*}
and the result follows.
\end{proof}

According to \cite[(8.22.1)]{szego}, we have 
\begin{align} \label{glpasymp}
L_n^{(\alpha)}(x) \sim \frac{1}{\sqrt{\pi}}  \, \frac{n^{\alpha/2 - 1/4}}{x^{\alpha/2 + 1/4}} \, e^{x/2} \, \sin\left(2\sqrt{nx} - \frac{\alpha \pi}{2} + \frac{\pi}{4} \right).
\end{align}
We make the following definition for notational convenience:
\[
\A \ddef \frac{\pi}{2} \left(\alpha - \frac{1}{2} \right).
\]
The expressions (\ref{glpasymp}) and (\ref{leading}) allow us to replace (\ref{der_form}) with 
\begin{align*}
&\frac{\dd}{\dd x}\,\Pi_n^{(\alpha)}(S)(x) \sim\\
&\frac{n^{(1-2\alpha)/4} \, (n-1)^{(1+2\alpha)/4} \, \exp((x-1)/2)}{\pi \, x^{(2\alpha + 3)/4}\,(x-1)} \ \times \\
&\left( 
\sin \left( 2\sqrt{n} - \A\right)\,\sin \left( 2\sqrt{(n-1)x} - \A \right)
-
\sin \left( 2\sqrt{n-1} - \A\right)\,\sin \left( 2\sqrt{nx} - \A \right)
\right).
\end{align*}

\subsection{Trigonometric Calculus} Using elementary properties of the sine and cosine functions we can turn the expression
\begin{align} \label{initial_trig}
\sin \left( 2\sqrt{n} - \A\right)\,\sin \left( 2\sqrt{(n-1)x} - \A \right)
-
\sin \left( 2\sqrt{n-1} - \A\right)\,\sin \left( 2\sqrt{nx} - \A \right)
\end{align}
into a more useful form for analysis. We apply the formul\ae
\begin{align*}
\sin A \, \sin B &= \frac{1}{2} \left( \cos(A-B)  - \cos(A+B) \right),\text{ and}\\
\cos A + \cos B &= 2 \cos((A-B)/2)\,\cos((A+B)/2)
\end{align*}
repeatedly to replace (\ref{initial_trig}) with the equivalent expression
\begin{align*}
&\cos \left( (\sqrt{n} + \sqrt{n-1})\sqrt{x} + (\sqrt{n-1} - \sqrt{n}) -\A \right)\,\cos \left( (\sqrt{n} - \sqrt{n-1})\sqrt{x} + (\sqrt{n} + \sqrt{n-1}) -\A \right) \\
&-\\
&\cos \left( (\sqrt{n} + \sqrt{n-1})\sqrt{x} + (\sqrt{n} - \sqrt{n-1}) -\A \right)\,\cos \left( (\sqrt{n} - \sqrt{n-1})\sqrt{x} - (\sqrt{n} + \sqrt{n-1}) -\A \right).
\end{align*}

Now we make two standard approximations:
\begin{align*}
\sqrt{n} \pm \sqrt{n-1} &= (2\sqrt{n})^{\pm1} + O(n^{-3/2})
\end{align*}
and apply them to the above expression to get
\begin{align*}
&\cos \left(  ( 2\sqrt{nx} - \A) - \frac{1}{2\sqrt{n}}\right) \cos \left(  ( 2\sqrt{n} - \A) + \frac{\sqrt{x}}{2\sqrt{n}}\right)\\
&- \\
&\cos \left(  ( 2\sqrt{nx} - \A) + \frac{1}{2\sqrt{n}}\right) \cos \left(  ( 2\sqrt{n} - \A) - \frac{\sqrt{x}}{2\sqrt{n}}\right).
\end{align*}
Next, we perform two operations: first we apply the sum of angles formula for each cosine term and second we make the substitutions 
\begin{align*}
&\cos \left( \frac{1}{2\sqrt{n}}\right) = 1 \text{ and } \cos \left( \frac{x}{2\sqrt{n}}\right) = 1, \text{ and}\\
&\sin \left( \frac{1}{2\sqrt{n}}\right) = \frac{1}{2\sqrt{n}} \text{ and } \sin \left( \frac{x}{2\sqrt{n}}\right) = \frac{\sqrt{x}}{2\sqrt{n}};
\end{align*}
the former is $O(1/n)$ and the latter is $O(1/n^{3/2})$.   These two operations allow us to replace the previous expression with
\begin{align*}
&\left(\cos \left( 2\sqrt{nx} - \A\right) + \frac{\sin \left( 2\sqrt{nx} - \A \right)}{2\sqrt{n}} \right)
\left(\cos \left( 2\sqrt{n} - \A\right) - \frac{\sqrt{x}\sin \left( 2\sqrt{n} - \A \right)}{2\sqrt{n}} \right) \\
&- \\
&\left(\cos \left( 2\sqrt{nx} - \A\right) - \frac{\sin \left( 2\sqrt{nx} - \A \right)}{2\sqrt{n}} \right)
\left(\cos \left( 2\sqrt{n} - \A\right) + \frac{\sqrt{x}\sin \left( 2\sqrt{n} - \A \right)}{2\sqrt{n}} \right).
\end{align*}

Multiplying these expressions and simplifying yields
\[
\frac{1}{\sqrt{n}} \left( \sin(2\sqrt{nx} - \A) \cos(2\sqrt{n}-\A) - \sqrt{x} \cos(2\sqrt{nx} - \A)\sin(2\sqrt{n} - \A)   \right),
\]
which is equal to 
\begin{align*}
&\frac{1}{\sqrt{n}} ( \sin(2\sqrt{nx} - \A) \cos(2\sqrt{n}-\A) - \cos(2\sqrt{nx} - \A)\sin(2\sqrt{n} - \A)  \\
& -(\sqrt{x}-1) \cos(2\sqrt{nx} - \A)\sin(2\sqrt{n} - \A)),
\end{align*}
which is furthermore equal to
\begin{align} \label{simplification}
\frac{1}{\sqrt{n}} \left( \sin(2\sqrt{n}(\sqrt{x} -1))-(\sqrt{x}-1) \cos(2\sqrt{nx} - \A)\sin(2\sqrt{n} - \A)\right).
\end{align}
Multiplying (\ref{simplification}) by 
\[
\frac{n^{(1-2\alpha)/4} \, (n-1)^{(1+2\alpha)/4} \, \exp((x+1)/2)}{\pi \, x^{(2\alpha + 3)/4}\,(x-1)},
\]
finally yields 
\begin{align} \label{final_with_x}
\frac{\dd}{\dd x}\,\Pi_n^{(\alpha)}(S)(x) &\sim \frac{\left(1-\frac{1}{n}\right)^{(1+2\alpha)/4} e^{\frac{x-1}{2}}}{\pi \, x^{(2\alpha + 3)/4}\,(x-1)} \times \\
&\left( \sin(2\sqrt{n}(\sqrt{x} -1))-(\sqrt{x}-1) \cos(2\sqrt{nx} - \A)\sin(2\sqrt{n} - \A)\right).
\end{align}
This is the expression that we will use for the final portion of our analysis.

\subsection{Completion of the Proof} Using the substitution $u = 2\sqrt{n}(\sqrt{x} -1)$, we compute:
\begin{align*}
\frac{\dd x}{x-1} &= 2 \, \left(\frac{u+2\sqrt{n}}{u(u+4\sqrt{n})}\right) \to \frac{1}{u},\text{ and} \\
\frac{x-1}{2} &= \frac{u(u+4\sqrt{n})}{8n} \to 0, \text{ and}\\
x^{(2\alpha + 3)/4} &= \left(\frac{(u+2\sqrt{n})^2}{4n} \right)^{(2\alpha+3)/4} \to 1.
\end{align*}
We also clearly have $(1-1/n)^{(2\alpha+3)/4} \to 1$ as $n \to \infty$, for all $\alpha$. 

Set $u_{n,\pm}^{(\alpha)} \ddef 2\sqrt{n}(\sqrt{x_{n,\pm}^{(\alpha)}} -1)$. As $x_{n,\pm}^{(\alpha)}$ approximate the first roots of $\sin (2\sqrt{n}(\sqrt{x} -1))$ before and after $x=1$, the $u_{n,\pm}^{(\alpha)}$ approximate the first roots of $\sin u$ before and after $u=0$, hence $u_{n,\pm}^{(\alpha)} \to \pm \pi$ for any fixed $\alpha$.  

Therefore, taking $n \to \infty$, applying the above approximation, and integrating (\ref{final_with_x}) yields
\begin{align*}
\Pi_n^{(\alpha)}(x_{n,+}^{(\alpha)}) - \Pi_n^{(\alpha)}(x_{n,-}^{(\alpha)}) &\to \frac{1}{\pi}\,\int_{-\pi}^\pi \frac{\sin u}{u}\,\dd u \, - \, 
\frac{\sin(2\sqrt{n} - \A)}{\pi\,2\sqrt{n}} \int_{-\pi}^{\pi} \cos(u+2\sqrt{n}-\A)\,\dd u \\
&\to \frac{1}{\pi}\,\int_{-\pi}^\pi \frac{\sin u}{u}\,\dd u  = \gamma,
\end{align*}
completing the proof of Theorem \ref{mainthm}.

\section{Numerical Justification} \label{numerics}

\subsection{Critical Points}  In the course of proving Theorem \ref{mainthm}, we made many asymptotic substitutions.  In this final section we choose a particularly convenient value of $\alpha$ and give some data illustrating the slowness of the convergence of the overshoot to $\gamma$. 

Setting $\alpha = -1/2$ gives us the exact formula 
\begin{align}
\frac{\dd}{\dd x} \Pi_n^{(-1/2)}(S)(x) = d_n^{(-1/2)} \, \left( \frac{-L_{n-1}^{(1/2)}(1)\, 
L_n^{(1/2)}(x) + L_{n-1}^{(1/2)}(x)\,L_n^{(1/2)}(1)}{x-1} \right),
\end{align}
hence 
\begin{align} 
&\Pi_n^{(-1/2)}(x_{n,+}^{(-1/2)}) - \Pi_n^{(-1/2)}(x_{n,-}^{(-1/2)})  \\
&= \frac{n!}{e\, \Gamma(n+1/2)} \, \int_{x_{n,-}^{(-1/2)}}^{x_{n,+}^{(-1/2)}} \frac{-L_{n-1}^{(1/2)}(1)\, 
L_n^{(1/2)}(x) + L_{n-1}^{(1/2)}(x)\,L_n^{(1/2)}(1)}{x-1} \, \dd x. \label{special_integral}
\end{align}
The choice of $\alpha = -1/2$ means that the asymptotic (\ref{glpasymp}) simplifies to 
\[
L_m^{(1/2)}(x)  \sim \frac{1}{\sqrt{\pi x}}  \, e^{x/2} \, \sin\left(2\sqrt{nx}\right)
\]
and $d_{n}^{(-1/2)} \sim \sqrt{n}$. However, for the purposes of numerical investigation, we will not apply these formul\ae, instead choosing to work with the polynomials explicitly.

Just as in the case of general $\alpha$, the numbers $x_{n,+}^{(-1/2)}$ and ${x_{n,-}^{(-1/2)}}$ are roots of a polynomial, hence can be estimated to high precision, as can the Gamma function $\Gamma (n+1/2)$.  And the integrand of (\ref{special_integral}) is a polynomial, hence is amenable to integration with a computer algebra package.

Using the computer algebra system \href{https://pari.math.u-bordeaux.fr/}{\texttt{gp-pari}} \cite{pari} we first define the integrand
\begin{center}
\texttt{LL(n,a,x) = (L(n-1,a+1,x)*L(n,a+1,1)-L(n-1,a+1,1)*L(n,a+1,x))/(x-1)}
\end{center}
and then the antiderivative of this polynomial:
\begin{center}
\texttt{I(n,a,x) = intformal(LL(n,a,x))}
\end{center}
which is, again, another polynomial. Up to this point, there are no simplifications or approximations.

Now we use two built-in routines for estimation: the \texttt{polroots} command, which returns a vector of approximations of a polynomial's roots, and the \texttt{gamma} command, which approximates the Gamma function.  Now we set
\begin{itemize}
\item numerical precision to 300 significant digits using \texttt{$\mathtt{\backslash}$p 300}
\item $n=200$
\end{itemize}
At this level of precision, using the \texttt{polroots} command we estimate our roots to be
\begin{align*}
x_{200,+}^{(-1/2)} \sim~ &1.23468886080318246205175640076637647798323949845707164146855\\
&181967104846688402531629527949568958184725501099221878799937\\
&065607066281702792065904979826237243575329891028544691448117\\
&4201139582708536428948881791368482448438408329248865333917731\\
&41854062545324039789475636307588151666472440252073636724439,\text{ and} \\[10pt]
x_{200,-}^{(-1/2)}\sim~&0.7900542198210110735737933107310435096052638050001253395683377\\
&0261682371463909575893541126569621731715791818914620650296428\\
&3864217580809768632697488288862041562201358460410751672195826\\
&5866620088384634191041289583048943312518954134983229273045299\\
&31039866409623037748110975121280531459368332355928748953.
\end{align*}
Now we evaluate \texttt{LL(n,-1/2,x)} at each of these numbers, take the difference, and multiply by $200!/e\,\Gamma(200.5)$ which yields the following approximation of the overshoot (truncated to show the slowness of the convergence):
\[
\gamma \sim 1.1808\dots,
\]
which is within 0.002 of the exact value of $\gamma$.  We also check that via the substitution 
\[
u = 2\sqrt{n}(\sqrt{x} -1)
\]
$x_{200,+}^{(-1/2)}$ and $x_{200,-}^{(-1/2)}$ are mapped to (again, truncated to show the slowness of convergence)
\[
u_{200,+}^{(-1/2)} \sim 3.1442 \dots,\text{ and } u_{200,-}^{(-1/2)} \sim -3.1438\dots,
\]
respectively, which are within 0.003 of $\pi$ and $-\pi$, respectively. 

\subsection{Convergence at $x=1$} In the case of the Hermite polynomials which are even functions in even degree and odd functions in odd degree, we were in the agreeable situation where the function $\sgn$ is odd, with symmetry about the origin, and the odd-degree Hermite polynomials all pass through the origin.  This is not the case with the GLP; it is routine to check that $L_n^{(\alpha)}(1) \ne 1/2$ for general $\alpha$ and all $n$.  Nonetheless, we make the following conjecture.

\begin{conj} \label{conj}
For all $\alpha > -1$ we have 
\[
\lim_{n \to \infty} \Pi_n^{(\alpha)}(1) = 1/2.
\]
\end{conj}

\noindent According to \cite[(4)]{carlitz}, we have
\[
L_n^{(\alpha)}(x)L_{n-1}^{(\alpha+1)}(x) = \frac{\Gamma(1+\alpha + n)}{2^{2n}n!} \, \sum_{r=0}^{n} \frac{(2r)!(2n-2r)!}{r!((n-r)!)^2\Gamma(1+\alpha + r)} L_{2r-1}^{(2\alpha+1)}(2x).
\]
Evaluating at $x=1$, we can combine this with (\ref{sum}) to get
\begin{align*}
\Pi_n^{(\alpha)}(S)(1) &=\frac{\Gamma(\alpha+1,1)}{\Gamma(\alpha + 1)}  -\frac{1}{e} \sum_{j=1}^n \frac{L_{j-1}^{(\alpha+1)}(1)}{\| L_{j-1}^{(\alpha + 1)} \|^2} L_{j}^{(\alpha)}(1) \\
&=\frac{\Gamma(\alpha+1,1)}{\Gamma(\alpha + 1)}  -\frac{1}{e} \sum_{j=1}^n \frac{\Gamma(1+\alpha + j)}{2^{2j}j! \| L_{j-1}^{(\alpha + 1)} \|^2} \sum_{r=0}^{j} \frac{(2r)!(2j-2r)!}{r!((j-r)!)^2\Gamma(1+\alpha + r)} L_{2r-1}^{(2\alpha+1)}(2).
\end{align*}
We can evaluate this expression for a few values of $\alpha$ and $n$:
\begin{center}
\begin{tabular}{|r|c|l|}
\hline
$n$ & $\alpha$ & $\Pi_n^{(\alpha)}(1)$ \\
\hline
100 & 0 & $0.4973032559\dots$ \\
1000& 0 & $0.4994002364\dots$ \\
\hline
100 & 1 & $0.5039460855\dots$ \\
1000& 1 & $0.5004991579\dots$\\
\hline
100 & 2 & $0.5256199161\dots$ \\
1000& 2 & $0.5096998025\dots$ \\
\hline

\end{tabular}
\end{center}
Since $\Pi_n^{(\alpha)}(x)$ is continuous in $\alpha$ and the Gamma function smoothly interpolates the factorial function, it would suffice to prove Conjecture \ref{conj} for integral values of $\alpha$. 

To finish the paper, we highlight one case in particular.  If $\alpha =0$, we have $\Gamma(1,1)/\Gamma(1) = 1/e$.  Using the identity $\| L_n^{(\alpha)} \|^2 = \Gamma(n+\alpha+1)/n!$, we can write
\[
\Pi_n^{(0)}(1) = \frac{1}{e}\left(1 - \sum_{j=1}^n \frac{1}{2^{2j}\,j} \sum_{r=0}^{j} \binom{2r}{r}\binom{2j-2r}{j-r} L_{2r-1}^{(1)}(2)\right).
\]
Noting that when $r=0$ we have $L_{2r-1}^{(1)}(2)=0$, our conjecture in the case $\alpha=0$ reduces to
\[
\sum_{j=1}^\infty \sum_{r=1}^j \sum_{k=0}^{2r-1} \frac{1}{2^{2j}\ j} \binom{2r}{r}\binom{2j-2r}{j-r} \binom{2r}{2r-1-k} \frac{(-2)^k}{k!} = 1 - \frac{e}{2}.
\]

\end{document}